\documentclass[times]{acsauth}

\usepackage{moreverb}

\usepackage[colorlinks,bookmarksopen,bookmarksnumbered,citecolor=red,urlcolor=red]{hyperref}

\newcommand\BibTeX{{\rmfamily B\kern-.05em \textsc{i\kern-.025em b}\kern-.08em
T\kern-.1667em\lower.7ex\hbox{E}\kern-.125emX}}

\def\volumeyear{2020}

\newtheorem{remark}{Remark}

\newtheorem{assumption}{Assumption}

\newtheorem{proposition}{Proposition}

\def\begequarr{\begin{eqnarray}}
\def\endequarr{\end{eqnarray}}
\def\begequarrs{\begin{eqnarray*}}
\def\endequarrs{\end{eqnarray*}}
\def\begarr{\begin{array}}
\def\endarr{\end{array}}
\def\begequ{\begin{equation}}
\def\endequ{\end{equation}}
\def\lab{\label}
\def\begdes{\begin{description}}
\def\enddes{\end{description}}
\def\begenu{\begin{enumerate}}
\def\begite{\begin{itemize}}
\def\endite{\end{itemize}}
\def\endenu{\end{enumerate}}

\def\lef[{\left[\begin{array}}
\def\rig]{\end{array}\right]}
\def\qed{\hfill$\Box \Box \Box$}
\def\begcen{\begin{center}}
\def\endcen{\end{center}}
\def\begrem{\begin{remark}\rm}
\def\endrem{\end{remark}}

\def\begmat#1{\begin{bmatrix}#1\end{bmatrix}}
\def\begali#1{\begin{align}{#1}\end{align}}
\def\begalis#1{\begin{align*}{#1}\end{align*}}

\def\cale{{\cal E}}

\def\calh{{\cal H}}
\def\calc{{\cal C}}

\def\calm{{\cal M}}
\def\calr{{\cal R}}
\def\calj{{\cal J}}

\def\cale{{\cal E}}

\def\call{{\cal L}}

\def\cala{{\cal A}}
\def\calj{{\cal J}}

\def\bq{\ell}

\def\L2e{{\cal L}_{2e}}

\def\rea{\mathbb{R}}

\def\diag{\mbox{diag}}

\def\col{\mbox{col}}
\def\hal{{1 \over 2}}

\def\diag{\mbox{diag}}
\def\rank{\mbox{rank}}

\def\bbs{{\mathbb S}}

\def\bfx{{\mathbf x}}

\def\bq{{\mathbf q}}

\def\L2e{{\cal L}_{2e}}

\def\rea{\mathbb{R}}

\def\diag{\mbox{diag}}

\def\col{\mbox{col}}
\def\hal{{1 \over 2}}

\def\diag{\mbox{diag}}

\def\bfqy{{\bf q}_y}

\def\yn{{\bf q}_N}
\def\bfqn{{\bf q}_N}
\def\dotbfqy{\dot{\bf q}_y}
\def\dotbfqn{\dot{\bf q}_N}

\def\IJRNLC{{\it Int. J. of Robust and Nonlinear Control}}
\def\TAC{{\it IEEE Trans. Automatic Control}}

\def\SCL{{\it Systems \& Control Letters}}
\def\AUT{{\it Automatica}}

\def\CST{{\it IEEE Trans. Control Systems Technology}}

\def\SIAM{{\it SIAM J. Control and Optimization}}

\usepackage{color}

\begin{document}

\runningheads{B. Yi~\MakeLowercase{\it et al}.}{Path following via I\&I orbital stabilization}

\title{Path following of a class of underactuated mechanical systems via immersion and invariance-based orbital stabilization\footnotemark[2]}

\author{Bowen Yi$^1$\corrauth, Romeo Ortega$^{2,3}$, Ian R. Manchester$^1$, Houria Siguerdidjane$^2$}

\address{\center
{\rm 1.} Australian Centre for Field Robotics \& Sydney Institute for Robotics and Intelligent Systems,\\
The University of Sydney, Sydney, NSW 2006, Australia
\\
{\rm 2.} Laboratoire des Signaux et Syst\`emes, CNRS-CentraleSup\'elec, Gif-sur-Yvette 91192, France
\\
{\rm 3.} Department of Control Systems and Informatics, ITMO University, St. Petersburg 197101, Russia
}

\corraddr{Australian Centre for Field Robotics \& Sydney Institute for Robotics and Intelligent Systems, The University of Sydney, NSW 2006, Australia (\texttt{bowen.yi@sydney.edu.au})}

\begin{abstract}
This paper aims to provide a new problem formulation of path following for mechanical systems without time parameterization nor guidance laws, namely, we express the control objective as an \emph{orbital stabilization} problem. It is shown that, it is possible to adapt the  immersion and invariance technique to design \emph{static} state-feedback controllers that solve the problem. In particular, we select the target dynamics adopting the recently introduced Mexican sombrero energy assignment method. To demonstrate the effectiveness of the proposed method we apply it to control underactuated marine surface vessels.
\end{abstract}


\footnotetext[2]{This paper is supported by the Australian Research Council, and by the Government of the Russian Federation (074U01), the Ministry of Education and Science of Russian Federation.}

\maketitle

\section{Introduction}
Tracking a geometric path in mechanical systems is a task often encountered in various application fields, \emph{e.g.}, robots, aerospace and autonomous vehicles. There are two approaches to formulate this problem, depending on whether the predefined path is parameterized by time or not, and they are known as \emph{trajectory tracking} or \emph{path following}, respectively. The former approach largely dominates the control literature and it has been extensively studied. Unfortunately, for \emph{non-minimum phase} systems, that are common in these applications, the achievable performance has a fundamental limitation, first identified in \cite{AGUetal}. It is possible to circumvent this limitation removing the time parameterizaton of the geometric path, leading to the path following approach---where the reparameterization of the time provides an additional degree of freedom to assign the evolution along the path and achieve the desired objective.

Two approaches to solve the path following problem (PFP) have been reported in the literature. The first approach controls the evolution of the reference point on the desired path, with the aid of \emph{guidance laws} that, in essence, allow us to reformulate the problem as a tracking control task. This technical route dominates the publications on the topic of path following. Some guidance laws, \emph{e.g.}, pursuit guidance, light-of-sight and virtual targets, have proven successful in various kinds of applications \cite{FOS}. The second class of methods, which includes the generation of virtual holonomic constraints and the principle of transverse feedback linearization, are based on \emph{set stabilization} \cite{CONetal,HLAetal,MORetal}. In these works, the dynamics are partitioned into transversal dynamics and tangential dynamics, which are controlled separately via feedback linearization. Besides the obvious lack of robustness of feedback linearization, the resulting designs are intrinsically local due to the existence of a singularity in the input matrix of the transversal dynamics.

In this paper, we reformulate the PFP adopting an \emph{orbital stabilization} approach, which is a particular kind of set stabilization. The main contributions of the paper are twofold.

\begin{itemize}
\setlength{\itemsep}{7pt}
  \item[{\bf C1}] The new reformulation of the PFP is solved proposing an extension to the recently developed \emph{Immersion and Invariance} (I\&I) orbital stabilisation method reported in \cite{ORTrnc} and using the Mexican sombrero energy-shaping construction of \cite{YIetal2} to solve the I\&I problem. The resulting controller is a globally defined static state-feedback that, without the appeal of any guidance laws, guarantees almost global convergence.
  \item[{\bf C2}] The result is applied to the benchmark problem of underactuated marine surface vessels. It is shown that the availability of some free functions in the feedback law allows us to ``shape'' the dynamical behaviour, playing a similar role as guidance laws.
\end{itemize}

The paper is organized as follows. The formulation of the PFP and its classical approach are recalled in Section \ref{sec2}. Section \ref{sec3} gives an I\&I formulation of orbital stabilization. In Section \ref{sec4} we select a target dynamics that is suitable for the PFP. The main result, namely, a constructive procedure to solve the PFP, is presented in Section \ref{sec5}. This is followed by a case study of marine surface vessels in Section \ref{sec6}. The paper is wrapped up with some concluding remarks in Section \ref{sec7}. To enhance readability, some of the proofs are given in appendices at the end of the paper.

 \ \\ \
\textbf{Notations.} $I_n$ is the $n \times n$ identity matrix. $\bbs$ denotes the unit circle, that is, $\bbs:= \rea$ mod $2\pi$. For $x\in\rea^n$, $A \in \rea^{n \times n}$ positive definite and a set $\cale \subset \rea^n$, we denote $|x|^2= x^\top x$, $\|x\|^2_A = x^\top Ax$, and $\|x\|_\cale:= \inf_{y\in\cale} |x-y|$. All mappings are assumed smooth enough. Given $f: \rea^n \to \rea$ we define the differential operator $\nabla f := \left({\partial f / \partial x}\right)^\top$. For a tall, full rank, matrix $B \in \rea^{n\times m}$, we define $B^\perp \in \rea^{(n-m) \times n}$ as a full-rank left-annihilator of $B$ and $B^\dagger$ as its Moore-Penrose pseudoinverse. When clear from the context the arguments of the mappings are omitted.

%
\section{Problem Formulation and Classical Approach}
\label{sec2}
%
\subsection{The path following problem}
\label{subsec21}
We consider $n$ degrees-of-freedom mechanical systems described in port-Hamiltonian (pH) form as\footnote{To avoid cluttering, and with some obvious abuse of notation, we sometimes mix $x$ and $(q,p)$.}
\begin{equation}\label{pH}
  \dot x
  =
  \begmat{~ 0_{n\times n} & A(q) ~~\\ ~ - A^\top(q) & -R(x)~~} \nabla H(x)
  +
  \begmat{~0_{n \times m}~\\ ~G(q)~}u,
\end{equation}
where $x:=\col(q,p)$, with $q \in \rea^n, p \in \rea^n$ the generalized configuration state and its momenta, respectively, $u \in \rea^m$ is the control input with $2\le m <n$, the matrix $R: \rea^{2n} \to \rea^{n\times n}$ satisfies $R(x) + R^\top(x) \ge 0$, the input matrix $G: \rea^n \to \rea^{n\times m}$ and $A: \rea^{n} \to \rea^{n\times n}$ are full rank. The  total energy function $H: \rea^n \times \rea^n \to \rea$ is given by
$$
H(q,p) := \hal p^\top M^{-1}(q) p + U(q)
$$
with the generalized inertia matrix $M: \rea^n \to \rea^{n\times n}_{>0}$ and $U:\rea^n \to \rea$ the potential energy function.

In order to be able to cover some important practical examples, depending on the selection of the matrix  $A(q)$, the momenta $p$ may be defined in the inertia or the body-fixed frames. Particularly, $A(q) =I_n$ implies that $p$ is defined in the inertia frame. Another popular way of modeling is done in the body-fixed coordinate, see \emph{e.g.} \cite{FOS}, in which $A(q)$ is a rotation matrix in $\text{SO}(n)$, yielding a \emph{constant} matrix $M$.

The output of the system to be regulated is defined, via the mapping $h: \rea^{n} \to\rea^{n_y}$, as
\begin{equation}
\label{bfqy}
{\mathbf q}_y = h(q), \quad \rank \{\nabla h(q)\} =n_y,
\end{equation}
where $h(q)$ is determined by specific control tasks. To streamline the main underlying idea we select, without loss of generality,  a \emph{planar closed} path, that is, $n_y = 2$.\footnote{We underscore  that the proposed framework is also applicable to \emph{high-dimensional} closed, non-intersecting paths by selecting the output mapping $h(q)$ properly, a fact that will be discussed in Remark {\bf R12} below.}

To streamline the formulation of the PFP that we adopt in the paper we need the following preliminaries. We are interested in following a regular Jordan curve, which is a plane curve topologically equivalent to (an homeomorphic image of) the unit circle. The desired path $\calc$ has the length $L$ in the output space with smooth parameterization $\sigma: [0,L] \to \rea^2$ and its image is ${\tt Im}(\sigma) = \calc$. A Jordan curve has the following implicit form as
$$
\calc:= \{\bfqy \in {\mathcal X} ~|~ \Phi(\bfqy) =0 \},
$$
where $\Phi: {\mathcal X} \to \rea$ is a smooth function, with $\nabla \Phi\neq 0$ on the open set ${\mathcal X} \subset \rea^2$ containing the desired path $\calc$. In view of the smoothness assumption, the desired Jordan curve $\calc$ is diffeomorphic to the unit circle $\bbs$.

We need the following additional assumption.

\begin{assumption}\rm \label{ass:jordan_curve} Given the Jordan curve $\Phi(\bfqy) = 0$, with $\Phi$ a regular function, the set
\begequ
\label{Omega}
\Omega : =  \{\bfqy \in \rea^2 | \nabla \Phi(\bfqy) = 0\}
\endequ
has {only} a finite number of {\em isolated} points and it guarantees
$$
\Phi(a) \Phi(b) > 0, \quad \forall a,~b \in \Omega.
$$
\end{assumption}

\ \\ \
\textbf{Formulation of the PFP} \cite{AGUetal} For the system \eqref{pH} and a desired smooth path characterized by the Jordan curve $\Phi(\bfqy)=0$, design a controller that achieves:

\begin{itemize}
\item[\bf P1] \emph{boundedness}: the states $x$ are bounded for all $ x(0)$;
\item[\bf P2] \emph{convergence and invariance}: the states $x$ will converge to the path set
$$
\lim_{t\to\infty}\| x(t) \|_{\calc_x} =0,
$$
where we defined
\begequ
\lab{cx}
\calc_x:=\{x\in\rea^{2n} ~ | ~ \Phi(h(q))=0\},
\endequ
which  is an invariant set;
\item[\bf P3] \emph{forward motion}: $\dot{x} \neq 0$ for all $t\ge 0$ and $x\in \calc_x$.
\end{itemize}

We underscore here that \textbf{P3} is the key property distinguishing the path following problem from the \emph{set stabilization} problem.

\subsection{Conventional approach to the problem}
\label{sec22}
A well-known approach to solve the PFP is to translate it into \emph{tracking} using guidance laws or motion generators. This is motivated by the fact that the desired path $\calc$ can be parameterized as
$
\bfqy^d = \sigma(\theta),
$
with $\theta \in [0,L]$ and a mapping $\sigma :[0,L] \to \rea^2$. In the mainstream methods, the variable $\theta$, then, propagates according to the motion generator dynamics
\begequ
\dot{\theta} = \Gamma(\theta,x), \quad \theta(0) = \theta_0.
\endequ
In this way, the PFP is translated into a standard tracking problem, that is, designing a feedback law such that
$$
\lim_{t\to\infty} |\bfqy(t) - \sigma(\theta(t))| =0.
$$
The free mapping $\Gamma(\theta,x)$ is then regraded as an \emph{additional control input}.

{The approach described above is well-established and it, essentially, has become the practical standard. However, the technique has two shortcomings that are openly recognized in the literature. First, the propagation variable $\theta$ projected on the path may admit {\em infinite} solutions for a {\em closed} curve, making it more suitable for open curves \cite{FOSetal,SAM}. Second, it does not ensure invariance on the pre-defined path \cite{HLAetal}, {\em i.e.},
$$
\Phi(h(q(t_0)) )= 0 \quad \not \Rightarrow \quad
\Phi(h(q(t)))=0 , ~ \forall t> t_0.
$$
}
Overcoming these two drawbacks is one of the motivations of this paper.

%
\section{An I\&I Approach to Orbital Stabilization}
\label{sec3}
%
In this paper we solve the PFP formulating it as an orbital stabilization problem similar to the one studied in \cite{ORTrnc}---where an extension to the I\&I technique \cite{ASTORT} originally applied for equilibria stabilization was proposed. We recall that this problem is to generate stable periodic solutions for the nonlinear system
\begequ
\label{eq:NL}
 \dot{x} = f(x) + g(x) u,
\endequ
with state $x \in \rea^{2n}$, input $u \in \rea^m$, with $m < 2n$, and $g(x)$ full rank. More precisely, we aim at defining a mapping $\hat{u} : \rea^{2n} \to \rea^m$ such that the closed loop is orbitally attractive. That is, for all $t\ge 0$
\begalis{
\dot{X}(t)  & = f(X(t)) + g(X(t))\hat{u}(X(t))\\
X(t) & = X(t+T),
}
and the set defined by its associated closed orbit
$$
\{x \in \rea^{2n}| x(t)= X(t),\; t\in [0,T]\},
$$
is attractive.

To apply this approach to the PFP it is necessary to extend \cite[Proposition 1]{ORTrnc} to the case where the target dynamics---and, consequently, the mapping defining the immersion manifold---are {\em functions of the state} of the system to be controlled. For ease of reference we recall below the main result of \cite{ORTrnc}.

\begin{proposition}
\label{prop:I&I}\rm
Consider the system \eqref{eq:NL}. Assume we can find mappings
\begali{
 \beta :\rea^{2n}\times\rea^{2n-p} \to \rea^m \quad  \alpha: \rea^p \to \rea^p
\lab{iimap}
 \quad \pi: \rea^p \to \rea^{2n}   \quad \phi: \rea^{2n}\to \rea^{{2n}-p}
}
with $p < 2n$, such that the following assumptions hold.
\begin{itemize}
    \item[\bf A1] (Target oscillator) The dynamical system
    \begequ
    \label{tardyn}
     \dot{\xi} = \alpha(\xi)
    \endequ
has non-trivial, periodic solutions $\xi_\star (t)=\xi_\star (t+T),\;\forall t \geq 0$, which are parameterized by the initial conditions $\xi(0)$.
    \item[\bf A2] (Immersion condition) For all $\xi\in\rea^p$,
    \begequ
    \label{fbi}
     g^\perp(\pi(\xi)) \left[f(\pi(\xi)) - \nabla \pi^\top(\xi) \alpha(\xi)\right]=0.
    \endequ

    \item[\bf A3] (Implicit manifold condition) The following set identity holds
    \begequ
    \label{impman}
    \begin{aligned}
     \{x \in\rea^{2n} ~|~ \phi(x)=0\} =
     \{x \in\rea^{2n} ~|~ x=\pi(\xi)\}.
     \end{aligned}
    \endequ
    \item[\bf A4] (Attractivity and boundedness condition) All the trajectories of the (extended) system
    \begequ
    \label{closed-loop}
     \begin{aligned}
          \dot{z} & = \nabla \phi^\top(x) [f(x) + g(x)\beta(x,z)] \\
          \dot{x} & = f(x) + g(x)\beta(x,z)
     \end{aligned}
    \endequ
    with $z(0)=\phi(x(0))$ and the constraint
    $$
    \beta(\pi(\xi), 0) = c(\pi(\xi))
    $$
    where
\begin{equation}
\label{eq:c}
    c(\pi(\xi)) := g^\dagger (\pi(\xi)) [\nabla \pi^\top(\xi)\alpha(\xi) - f(\pi(\xi)) ],
\end{equation}
    are bounded and satisfy
    $
    \lim_{t\to\infty} z(t)=0.
    $
\end{itemize}
Then, the feedback law $u = \beta(x, \phi(x))$ ensures that $x_\star (t)=\pi(\xi_\star (t))$ is orbitally attractive.
\qed
\end{proposition}

Notice that \eqref{fbi} and \eqref{eq:c} are equivalent to
$$
f(\pi(\xi)) + g(\pi(\xi)) = \nabla \pi^\top(\xi)\alpha(\xi),
$$
which is the necessary and sufficient condition for invariance of the set \eqref{impman}. Also, as shown in \cite[Proposition 3]{WANetal} the attractivity condition \textbf{A4} can be replaced by a contraction condition---see also \cite[Corollary 2]{MANSLOtac}.

\begin{figure}[h]
	\begin{center}
	\includegraphics[width=0.48\textwidth]{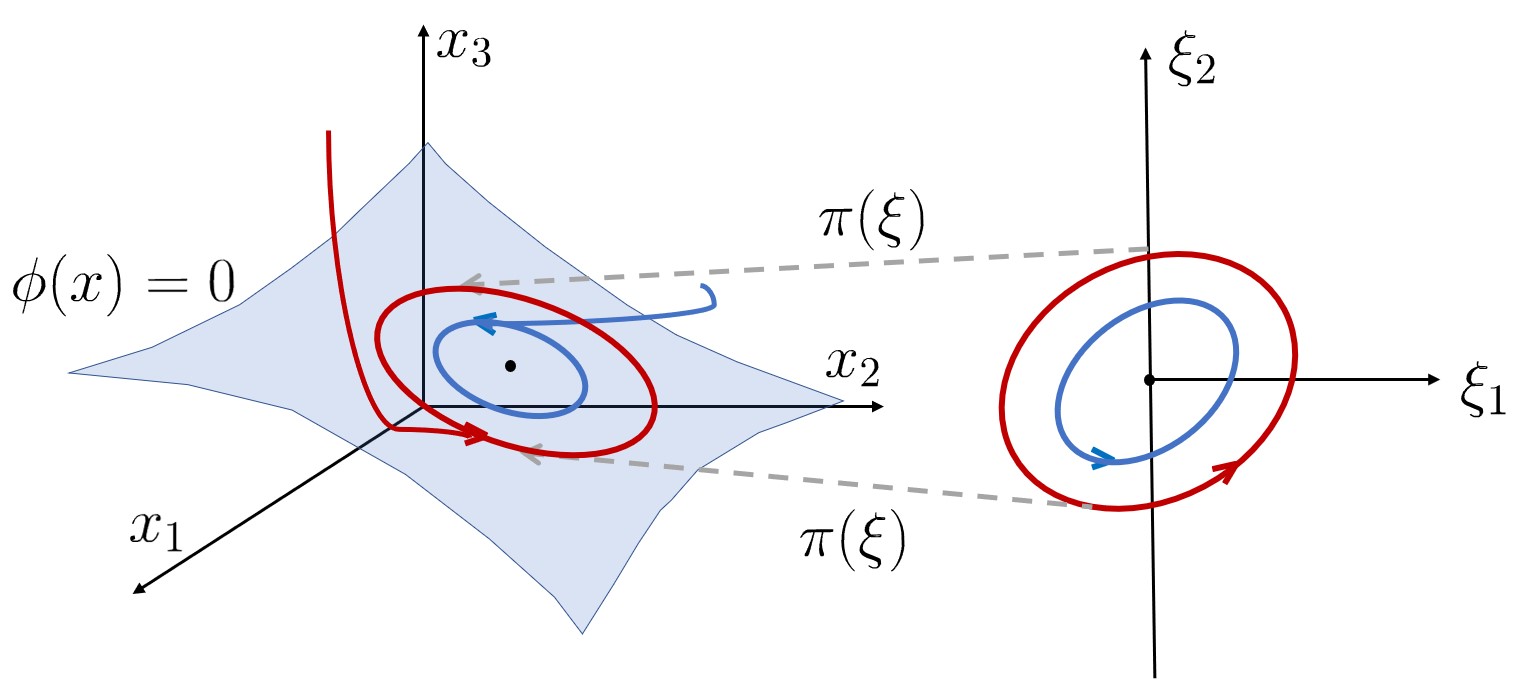}
          \includegraphics[width=0.25\textwidth]{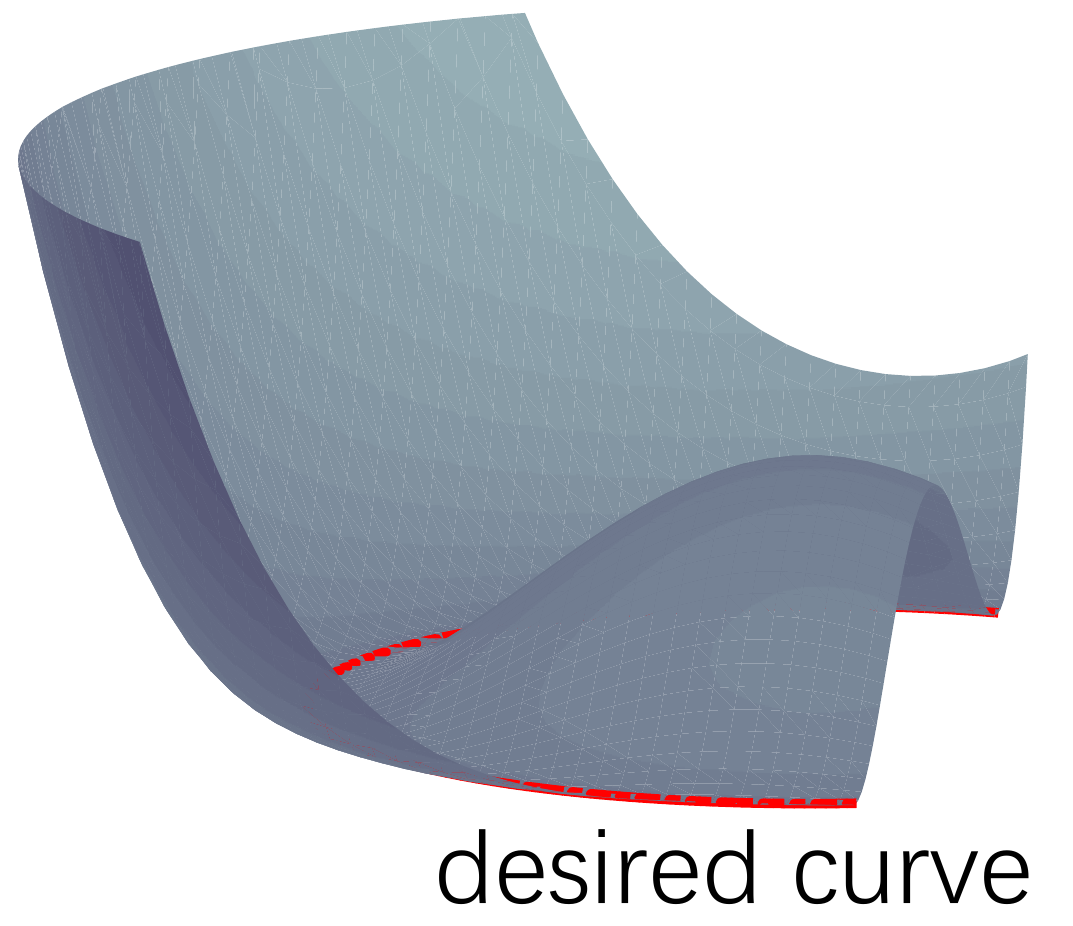}
	\caption{(a) Orbit in the target dynamics and mapping to the system state; (b) the Hamiltonian function with the target orbit.}
	\label{fig:ORTrnc}
  \end{center}
\end{figure}

%
\section{Definition of the Target Dynamics}
\label{sec4}
%
A first, key step for the application of the I\&I procedure to solve the PFP is the suitable selection of the target dynamics, which is carried-out in this section.
\subsection{Target dynamics oscillator}
\label{subsec41}
In order to obtain an attractive target oscillator, we utilize the Mexican sombrero energy assignment method proposed in \cite{YIetal2}, whose underlying mechanism is to assign an energy function that has a minimum at the desired closed curve, see Fig. \ref{fig:ORTrnc}(b). A similar idea, for potential energy shaping of {\em fully-actuated} mechanical systems, can be traced back to the work \cite{DUISTR}.

We endow the target dynamics with the port-Hamiltonian form
\begequ
\label{eq:target}
\dot{\xi} = F(\xi,x) \nabla V_d(\xi) =: \alpha(x,\xi)
\endequ
where $\xi \in \rea^2$ is the state of the target dynamics, $V_d:\rea^2 \to \rea$ plays the role of a desired potential energy, and $F:\rea^2 \times \rea^{2n} \to \rea^{2\times2}$ is a mapping to be selected satisfying $F(\xi,x) + F^\top(\xi,x) \leq 0$.

Notice that, as indicated above,  the target dynamics \eqref{eq:target} explicitly {\em depends on $x$}. As will be shown below, this is necessary  to deal with non-minimum phase systems and for performance enhancement in the PFP.

Similarly to the desired closed-loop dynamic of \cite{YIetal2} we propose
\begequ
\label{F}
F(\xi,x) =  \calj(\xi,x)- \calr(\xi),
\endequ
with
\begequ
\label{J}
\calj(\xi,x) := \begmat{ 0 & {{w(\xi,x)} \over\Phi(\xi)} \\ \\ -{{w(\xi,x)} \over\Phi(\xi)} & 0}
\endequ
where $w:\rea^2 \times \rea^{2n} \to \rea$ is a \emph{free} mapping and $\calr(\xi)>0$ and
\begequ
\label{eq:Vd}
V_d(\xi) = {1\over 2} |\Phi(\xi)|^2.
\endequ
\subsection{Stability properties of the target dynamics}
\label{subsec42}
%
Regarding the target dynamics oscillator \eqref{eq:target}-\eqref{J}, we have the following.

\begin{proposition}
\rm \label{prop:oscillator}
Consider the target dynamics \eqref{eq:target}-\eqref{J} verifying {Assumption \ref{ass:jordan_curve}}.\footnote{Assumption \ref{ass:jordan_curve} ensures that the set $\Omega$ is in the interior of $\calc_{\tt T}$.} Define the set
\begequ
\lab{ct}
\calc_{\tt T}:=\big\{\xi \in \rea^2 ~|~ \Phi(\xi)=0 \big\}.
\endequ
If the scalar function $w(\xi,x)$ is such that
\begin{equation}
\label{eq:cond_w}
w(\xi,x) \neq 0, \quad \forall \xi \in \mathcal{C}_{\tt T}, ~\forall x\in \rea^{2n},
\end{equation}
the set $\calc_{\tt T}$ is exponentially orbitally stable.

Furthermore, if the function $w(\cdot)$ depends only on $\xi$, namely, $w(\xi,x) := w_0(\xi)$, then the orbit $\calc_{\tt T}$ is \emph{almost globally} exponentially stable.
\end{proposition}
\begin{proof}\rm
From
$$
\dot{V}_d = - \|\nabla V_d\|^2_{\calr(\xi)} \le 0,
$$
and the fact that $\calr(\xi)$ is positive definite, we conclude that  $V_d \in \call_\infty$ and $\nabla V_d \in \call_2$. Now, using the fact that
\begequ
\lab{nabvd}
\nabla V_d(\xi)= \Phi(\xi) \nabla \Phi(\xi)
\endequ
and doing some simple calculation we get
$$
{d\over dt}(\nabla V_d) = - \Phi(\xi) (\nabla\Phi(\xi))^\top \calr(\xi) \nabla\Phi(\xi).
$$
Now
$$
V_d  \in \call_\infty
~\Leftrightarrow~
\Phi \in \call_\infty
~\Rightarrow ~
\xi \in \call_\infty,
$$
where the last implication follows from regularity of $\Phi$. Hence, we conclude that ${d\over dt}(\nabla V_d) \in \call_\infty$. Invoking Barbalat's Lemma we conclude that $\nabla V_d$ converges to zero, that is
$$
\lim_{t\to\infty} \Phi(\xi(t)) \nabla \Phi(\xi(t)) = 0.
$$
Thus, $\xi$ ultimately converges either to the curve $\calc_{\tt T}$ or to the {zero-Lebesgue measure set $\Omega$}.

The proof of the property with the qualifier ``exponentially" can be established following \cite{YIetal2}.

For the case $w(\xi,x)= w_0(\xi)$, the dynamics \eqref{eq:target} becomes an {\em autonomous} system. The claim regarding the almost global stability of $\calc_{\tt T}$ is established proving that there are no asymptotically stable equilibria in $\Omega$, and is presented below.

First, we show that all the points of $\Omega$ are in the interior of $\calc_{\tt T}$. Invoking the Jordan curve theorem \cite[Theorem 5.20, pp. 112]{ARM}, the complement of the desired path, $\rea^2/\calc_{\tt T}$, consists of exactly two connected components, one of which is bounded (the interior $\Omega_i$), and the other one is unbounded (the exterior $\Omega_e$). Besides, the curve $\calc_{\tt T}$ is the boundary of each component. Thus, the set $\Omega_i \cup \calc_{\tt T}$ is a compact set, and thus contains a least one extremum point, where $\nabla\Phi(\xi) =0$. Furthermore, this extremum point should be in $\Omega_i$ rather than in $\calc_{\tt T}$, since $\nabla \Phi(\xi)\big|_{\calc_{\tt T}} \neq 0$ due to the definition of Jordan curves. Assumption \ref{ass:jordan_curve} guarantees all the points in $\Omega$ are inside $\Omega_i$.

{Second, we prove that at least one point in $\Omega$ is unstable. Without loss of generality, we assume\footnote{Otherwise, we can select $-\Phi$ to carry out the analysis.} $\Phi(\xi)\big|_{\Omega_i} <0$ and $\Phi(\xi)\big|_{\Omega_e}>0$. Thus, there exists an isolated minimum $\xi_\star$ in $\Omega_i$, satisfying
$$
\nabla \Phi(\xi_\star) =0, \quad \nabla^2 \Phi(\xi_\star)>0.
$$
For a sufficiently small $\varepsilon>0$, we construct a Lyapunov-like function in the small neighborhood $B_\varepsilon(\xi_\star) \subset \Omega_i$ as
$$
W(\xi) = - \Phi(\xi) + \Phi(\xi_\star).
$$
 It is clear that
\begalis{
\dot{W}(\xi) & = - (\nabla \Phi^\top \calr(\xi) \nabla \Phi) W =: - R_0(\xi) W,
}
and
$$
\dot{W}(\xi) > 0, \quad \forall \xi \in B_\varepsilon(\xi_\star) / \xi_\star,
$$
where $R_0(\xi)$ is positive definite in $\rea^2 /\Omega$. According to Lyapunov's instability theorem \cite[Theorem 4.3]{KHA}, as well as the property of the function $W(\xi)$, we conclude that $\xi_\star$ is an unstable equilibrium.}

Finally, we will show that the other equilibrium points in $\Omega$ are unstable by contradiction. We assume that there exits another point $\xi_0 \in \Omega_i/\{\xi_\star\}$, which is asymptotically stable. According to the topological properties of the domain of attraction \cite[Proposition 5.44, pp. 217]{SAS}, the domain of attraction $\cale_0$ of $\xi_0$ is an \emph{open}, invariant set, the boundary of which, denoted as $\partial \cale_0$, is {\em invariant} as well. Noticing that $\xi_0$ is not connected to the Jordan curve $\calc_{\tt T}$ and that there is an unstable equilibrium in the interior $\Omega_i$, we conclude that the boundary $\partial\cale_0$ is in the set $\Omega_i$. We have already shown that all the trajectories will converge to $\calc_{\tt T} \cup \Omega$, that contradicts the invariance of $\partial \cale_0$. Therefore, we conclude that there is no asymptotically stable equilibrium in $\Omega_i$, and the $\omega$-limit set of the dynamics \eqref{eq:target} with $w(\xi,x) = w_0(\xi)$ contains $\calc_{\tt T}$ only.
\end{proof}
\subsection{Remarks}
\label{subsec43}
%

\noindent {\bf R1}
Even though $\Phi(\xi)$ appears in the denominator of the interconnection matrix $\calj(\xi,x)$ in \eqref{J}, the closed-loop dynamics \eqref{eq:target} is well-posed---see \cite{YIetal,YIetal2} for additional details.

\vspace{0.3cm}

\noindent {\bf R2}
The damping matrix $\calr(\xi)$ can be selected freely to assign the convergence speed to the orbit. The function $w(x,\xi)$ determines the behaviour on the orbit---equivalently the path following speed.\footnote{We will show in the proof of Proposition \ref{prop:main}, that the mapping $w(x,\xi)$ appears in the ``zero dynamics", thus with a suitable selection of this mapping we may have the possibility to treat non-minimum phase systems.}

\vspace{0.3cm}

\noindent {\bf R3}
If $\Phi(\xi)$ has an unique extremum globally, we can complete the proof, trivially, with the Lyapunov instability theorem, which is the case studied in \cite{YIetal2}.

\vspace{0.3cm}

\noindent {\bf R4}
A Jordan curve can be represented with an \emph{infinite} selection of functions. For instance, a well-known function for Cassini oval is
$$
\Phi(\xi) = \xi_1^4 + \xi_2^4 - 2a_0^2(\xi_1^2 - \xi_2^2) + a_0^4 - b_0^4,
$$
which is shown in Fig. \ref{fig:8} with $a_0=1,~b_0=1.2$. Notice that are three points satisfying $\nabla \Phi(\xi) = 0$. We also show in the figure the phase portrait of the associated target dynamics with $w(\xi,x) = 1$ and $\calr(\xi) = I_2$, where it is shown that the three equilibria are unstable.

\begin{figure}[h]
  \centering
  \includegraphics[width=0.8\textwidth]{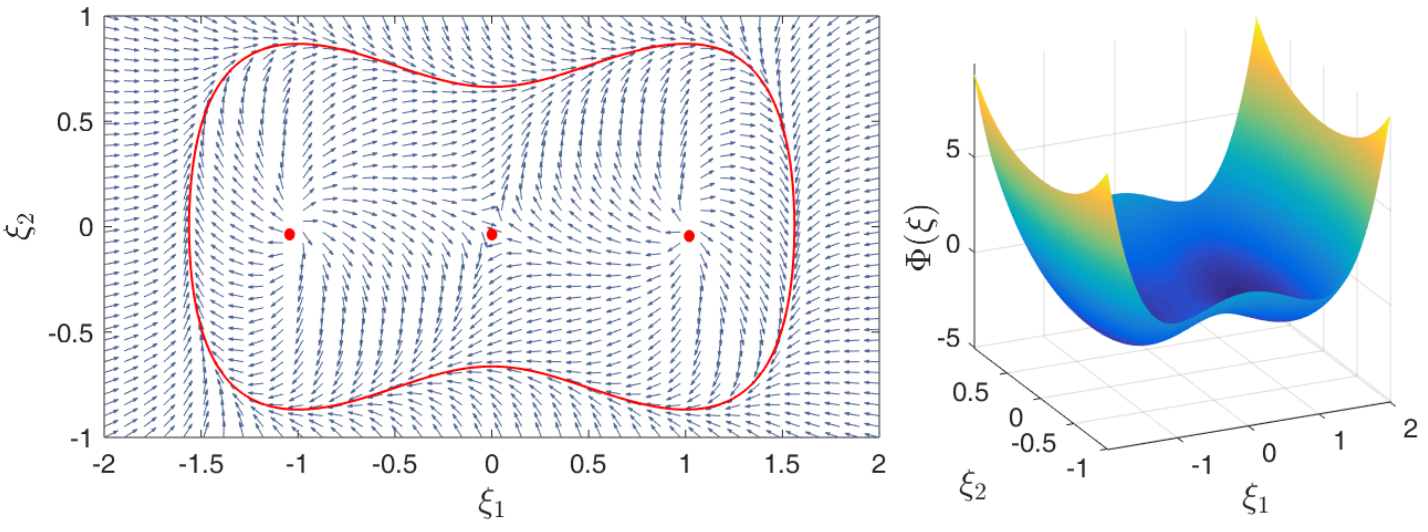}
  \caption{Illustration of the path of the Cassini oval and associated phase plane}\label{fig:8}
\end{figure}

\vspace{0.3cm}

\noindent {\bf R5}
The guiding vector-field for path following with kinematic models studied in the literature, {\em e.g.}, \cite{KAPetal}, shares some similarities with the proposed target dynamics. The authors also studied the possible existence of asymptotically stable equilibria in $\Omega$. However, this is only conjectured in  \cite[Hypothesis 1, pp. 1379]{KAPetal} and it is studied via the inspection of the phase portraits of some examples. Interestingly, following the proof of Proposition \ref{prop:oscillator}, we give an affirmative answer to this {conjecture}.  Indeed, we have proven that there is no asymptotically stable equilibrium point in $\Omega$. In \cite{WIL}, it has been proved that the domain of attraction of a limit cycle is also an \emph{open} set, a fact providing an alternative way to prove Proposition \ref{prop:oscillator}.

%
\section{Main Result}
\label{sec5}
%
Before presenting the main result let us take a brief respite and recapitulate our proposal of an I\&I formulation of the PFP for the plant dynamics \eqref{pH}. First, the control objective is to render the set  \eqref{cx}, that is,
$$
\{x\in\rea^{2n} ~ | ~ \Phi(h(q))=0\}
$$
attractive and invariant. Second, we defined a target dynamics \eqref{eq:target}-\eqref{J}
$$
\dot{\xi} = F(\xi,x) \nabla V_d(\xi)
$$
such that the set \eqref{ct}
$$
\big\{\xi \in \rea^2 ~|~ \Phi(\xi)=0 \big\}
$$
is orbitally exponentially stable. We proceed then to design now an I\&I control law that immerses the target oscillator \eqref{eq:target} into the plant dynamics \eqref{pH}, via the attractive and invariant manifold
\begin{equation}
\label{eq:M}
\calm := \{(q,p,\xi)\in\rea^{n}\times \rea^n \times \rea^2~ |~ h(q)= \xi \}.
\end{equation}
Following the I\&I procedure articulated in Proposition \ref{prop:I&I}, we will achieve this objective by ensuring the existence of a control $u=\beta(\xi,x)$  and a mapping $\pi(\xi,x)$ such that, on one hand, the manifold
$$
\{(x,\;\xi) \in \rea^{2n} \times \rea^2\; | \;x = \pi(\xi,x)\},
$$
is attractive and invariant. On the other hand, that this implies that $\calm$ in \eqref{eq:M} is also attractive and invariant.
\subsection{An additional assumption and a suitable change of coordinates}
\label{subsec51}
%
To carry out the controller design we need an additional assumption that is motivated by the following considerations. First, since we are interested in regulating the mapping $h(q)$, it is reasonable to view it as an ``output signal" and impose some conditions on its selection that will simplify this task. In particular, we would like to ensure that it has a well defined vector relative degree $(2,2)$---hence easily stabilizable via partial feedback linearization. Second, to carry out the calculations we will find convenience to introduce a partial change of coordinates for the configuration state $q$, hence it is necessary to ensure that it is a diffeomorphism. These conditions, are guaranteed imposing the following technical, but mild, assumption on the mapping $h(q)$.

\begin{assumption}
\label{ass:h}\rm
The output mapping $h(q)$ is such that the following is true.

\noindent {\bf (i)} The matrix\footnote{This is the input matrix after calculating ${d^2 \over dt^2} h(q)$.}
$$
\nabla h^\top(q) A(q) M^{-1}(q) G(q) \in \rea^{2\times {m}}
$$
is full rank.

\noindent{\bf (ii)} There exists a mapping $N:\rea^n \to \rea^{n-2}$ such that the following matrix\footnote{For future reference, we defined the vector $\bq$ and its corresponding partition, with ${\bf q}_y$  the output to be regulated with the remaining coordinates $(\bq_N,\dot \bq_N)$ bounded.}
\begequ
\lab{defbfq}
T(q)=\begmat{h(q) \\ N(q)}=:\bq=\begmat{\bfqy \\ \bfqn},
\endequ
is injective.
\qed
\end{assumption}

{The construction of a mapping $N(q)$ satisfying condition {\bf(ii)} is related to the so-called ``$\tilde{P}(m,n)$ Problem'' that has been studied in the context of state observer design in \cite[Section 9.2]{BER} and \cite{BERetal}. Indeed, the existence of $N(q)$ can be guaranteed in convex sets invoking Wazeski's theorem, see \cite[Theorem 9.4]{BER}.}

We write now the system \eqref{pH} in the coordinates $\bfx:=\col(\bq,p)$ yielding the port-Hamiltonian system
\begequ
\label{pH_trans}
\dot \bfx
=
\begmat{ 0 & \cala(\bq) \\ - \cala^\top(\bq) & - {\bf R}(\bfx)} \nabla \calh(\bfx)
+
\begmat{0 \\ {\bf G}(\bq)} u
\endequ
with the definitions
\begalis{
\calh(\bq,p)& := \hal p^\top M^{-1}( T^{\tt I}(\bq)) p + U(T^{\tt I}(\bq))\\
\cala(\bq)& := \nabla T^\top(T^{\tt I}(\bq))A(T^{\tt I}(\bq))\\
{\bf R}(\bq,p)& := R(T^{\tt I}(\bq),p)\\
{\bf G}(\bq))&:= G(T^{\tt I}(\bq)),
}
and $T^{\tt I}$ denoting the inverse mapping of $T$. Our motivation to introduce the new representation of the system is to simplify the computations leading to the verification of the conditions of Proposition \ref{prop:I&I}
of the I\&I procedure.
\subsection{Construction of the I\&I controller}
\label{subsec52}
%
In this section we elaborate on three essential differences of the design of the I\&I controller for PFP with respect to the design for classical regulation \cite{ASTORT} or orbital stabilization \cite{ORTrnc} tasks.

\vspace{0.3cm}

\noindent {\bf D1} Due to the dependence of the target dynamics on the systems state, some modifications to the conditions of Proposition \ref{prop:I&I} are needed. A first condition is that, on the manifold $\calm$, the state $\xi$ can be expressed as an explicit function of the systems state---this invertibility property will be guaranteed verifying the classical Jacobian rank condition of the Implicit Function Theorem.  A second modification is the definition of new immersion and implicit manifold  assumptions, {\bf A2} and {\bf A3}, respectively. Finally, it is necessary to redefine the attractivity and boundedness condition  {\bf A4} for our particular problem.

\vspace{0.3cm}

\noindent {\bf D2}  Recalling that our final objective is to immerse the target oscillator \eqref{eq:target} into the plant dynamics \eqref{pH_trans}, via the attractive and invariant manifold $\calm$, given in \eqref{eq:M}, it is necessary to ensure that this is \emph{implied} by the attractivity and invariance of
\begequ
\lab{me}
\calm_e:=\{(\bfx,\;\xi) \in \rea^{2n} \times \rea^2\; | \;\bfx = \pi(\xi,\bfx)\}.
\endequ

\vspace{0.3cm}

\noindent {\bf D3}  In the I\&I controller for regulation of {Proposition \ref{prop:I&I}} it is, first, assumed the existence of the mappings \eqref{iimap} such that the conditions {\bf A1-A4} hold, and then a state-feedback control law that ensures the regulation objective is constructed. In PFP we give an explicit expression for the mappings $c$, $\pi$ and $\phi$ that ensure the (new versions of) immersion {\bf A2} and implicit manifold {\bf A3} conditions. Then, we assume the existence of the mappings $\beta$ and $w$---the first one being the effective state feedback to be applied and the latter defining the interconnection target dynamics via \eqref{J}---such that the manifold is attractive and the off-the-manifold coordinate and the closed-loop system trajectories are bounded, {\em i.e.}, verifying condition {\bf A4}, completing in this way the design.

\subsection{Construction of the I\&I controller}
In this subsection we finalize the design of the I\&I controller that solves the PFP. In view of the explanation given in {\bf D3} above, the main result, naturally, split in two parts, summarized in  {Propositions \ref{pro3}} and {\ref{prop:main}} below.

\begin{proposition}
\label{pro3}\rm
Consider the dynamics \eqref{pH_trans} and the target system \eqref{eq:target} verifying Assumptions \ref{ass:jordan_curve}-\ref{ass:h}. There exist mappings
$$
\pi: \rea^{2} \times \rea^{2n} \to \rea^{2n}, ~
\phi:\rea^{2n} \to \rea^{n}, ~
c :\rea^{2} \times \rea^{2n} \to \rea^{2n},
$$
with $\pi$ verifying
$$
\rank\{\nabla_\xi \pi (\xi,\bfx)\}=2,
$$
such that the system \eqref{pH_trans} in closed-loop with $u= c(\xi,\bfx)$ satisfies the following conditions.

\noindent {\bf A2'} (New immersion condition) The manifold $\calm_e$ \eqref{me} is invariant, equivalently,
\begin{equation}
\label{dot_bfx}
\dot{\bfx}|_{\bfx = \pi(\xi,\bfx)}  = {d \over dt}[\pi(\xi,\bfx) ].
\end{equation}
\noindent {\bf A3'} (New implicit manifold condition) The following equivalence holds\footnote{The implicit function is written with the argument of the original coordinate $x$, rather than $\bfx$, for convenience to obtain the feedback law.}
\begequ
\label{phi_x}
\phi(x) =0 \quad \Leftrightarrow \quad
\bfx = \pi(\bfx,\xi).
\endequ
Furthermore, the controller $u=c(\xi,\bfx)$ guarantees the manifold $\calm$ defined in \eqref{eq:M} is invariant.
\qed
\end{proposition}

The proof is given in Appendix \ref{sec:app_proof_lemma}.

To complete the I\&I design it remains to verify condition {\bf A4}. Towards this end, we split the tasks of ensuring convergence of the off-the-manifold coordinate $z=\phi(x)$ and establishing boundedness of the coordinates $(\yn, \dotbfqn)$. The first task is solved selecting the mapping $\beta(x,z)$ {\em as a function of} a free mapping $w(\xi,x)$. Then, the latter is selected to ensure the second boundedness condition. Both steps are summarized in the proposition below.

\begin{proposition}
\label{prop:main}\rm
Consider the dynamics \eqref{pH} verifying {Assumptions \ref{ass:jordan_curve}-\ref{ass:h}} and the following
\begin{itemize}
  \item[\bf T1] The origin of the dynamics
  \begequ
  \label{dyn_z}
  \begin{aligned}
  \dot{z} & = \nabla \phi^\top
  \left\{
  \begmat{0_{n \times n} & A \\ -A^\top & - R}
  \nabla H
  +
  \begmat{0 \\ G} \beta(x,z)
  \right\}
  \end{aligned}
  \endequ
  is asymptotically stable for all mappings $w(\xi,x)$, {with $z=\phi(x)$ defined in \eqref{app:phi_x}}.
  \item[\bf T2] The mapping $w(\xi,x)$ guarantees the boundedness of the coordinates $(\yn, \dotbfqn)$.
\end{itemize}
Under these conditions, the feedback law $u= \beta(x,\phi(x))$ ensures the properties \textbf{P1}-\textbf{P3} of the problem formulation, providing a solution of the PFP.
\end{proposition}
\begin{proof}\rm
We first verify property \textbf{P1} of boundedness of $x=\col(q,p)$. In view of the fact that $\nabla T$ is full rank the change of coordinates
$$
x ~ \mapsto ~ (\bfqy,\bfqn,\dotbfqy,\dotbfqn)
$$
is a diffeomorphism, thus the boundedness of $x$ is equivalent to that of $(\bfqy,\bfqn,\dotbfqy,\dotbfqn)$. Due to the assumption {\bf T2}, we only need to prove the boundedness of $(\bfqy,\dotbfqy)$, which will be done introducing yet another change of coordinate, namely
$$
(\bfqy,\dotbfqy) \mapsto (\bfqy,z).
$$
First, notice that the mapping $z=\phi(x)$ defined in \eqref{app:phi_x} can be written as
\begequ
\lab{z}
z = \dotbfqy - F(\bfqy,x)\nabla V_d(\bfqy)=:\mathcal{T}_x(\bfqy,\dotbfqy).
\endequ
Computing the Jacobian of this parameterized mapping yields
\begalis{
\nabla \mathcal{T}_x & = I - {\nabla_{\dotbfqy} [F(\bfqy,x)\nabla V_d(\bfqy)]}\\
&= I - {1\over \Phi(\bfqy)}\nabla V_d^\top \nabla_{\dotbfqy}w(\xi,x) \begmat{0 & 1 \\ -1 & 0}\\
& = \begmat{1 &- \nabla \Phi^\top \nabla_{\dotbfqy}w(\xi,x)  \\ \\ \nabla \Phi^\top \nabla_{\dotbfqy}w(\xi,x) & 1}
}
that, regardless of the selection of the mapping $w(\xi,x)$, is full rank.

Now, the assumption {\bf T1} ensures the convergence---and, consequently, the boundedness---of $z$, we thus only need to prove the boundedness of the state $\bfqy$, the dynamics of which is given in \eqref{z}
$$
\dotbfqy = F(\bfqy,x) \nabla V_d(\bfqy) + z,
$$
where $z$ is, in view of  assumption {\bf T1}, \emph{asymptotically} decaying to zero. Let us now compute
$$
\begin{aligned}
\dot{V}_d & = - \nabla V_d^\top \calr(\bfqy) \nabla V_d + \nabla V_d^\top z \\
& \le  - 2\lambda_{\min}\big\{ \calr(\bfqy)\big\} |\nabla \Phi|^2 V_d  + |z||\Phi \nabla \Phi |
\\
& \le - 2\Big(\lambda_{\min}\big\{ \calr (\bfqy)\big\} + r\Big) |\nabla \Phi|^2 V_d + {1\over 4r}|z|^2
\end{aligned}
$$
for some $ r \in (0, \lambda_{\min}\big\{ \calr (\bfqy)\big\})$.

From Assumption \ref{ass:jordan_curve}, we have $\nabla \Phi(\bfqy) =0$ only in the set $\Omega$ with finite isolated points. Thus, for any $\varepsilon >0$ there always exists $\ell_\varepsilon > 0$ such that
$$
|\nabla \Phi(\bfqy)|^2 >  {\ell_\varepsilon \over 2\lambda_{\min}\big\{\calr(\bfqy)\big\}},
$$
except some small neighborhoods of $\Omega$, then yielding
$$
\dot{V}_d \le - \ell_\varepsilon V_d + {1\over 4r}|z|^2.
$$
We conclude the boundedness of $V_d = \hal |\Phi(\bfqy)|^2$, as well as $\bfqy$, verifying \textbf{P1}.

We proceed now to prove that the convergence and invariance property \textbf{P2} is also satisfied. Assumption \textbf{T1} guarantees the set $\{ x\in\rea^{2n} ~|~ \phi(x) = 0\},$
which is equal to  the set
$$
\{ x\in\rea^{2n} ~|~ \Phi(\bq_y) = \Phi(h(q)) = 0\},
$$
is attractive and invariant. We write the dynamics of the transverse coordinate $(z,\Phi(y))$ to the desired orbit as follows,
$$
\begin{aligned}
\dot{z} & = \nabla\phi^\top \dot{x}\\
\dot{\Phi} & = - \|\nabla \Phi\|^2_{\calr(\bq_y)} \Phi + \nabla \Phi^\top z.
\end{aligned}
$$
The latter equation is an exponentially stable (scalar) linear time-varying system perturbed by an asymptotically decaying term. Hence,
$$
\lim_{t\to\infty} |(z(t), \Phi(\bq_y(t)))| =0,
$$
validating \textbf{P2}.

Regarding \textbf{P3}, from \eqref{eq:cond_w} it is straightforward to see
$$
\dot{x} \neq 0, \quad \forall x\in \calc_x.
$$
This completes the proof.
\end{proof}
\subsection{Remarks}
\label{subsec53}
%
\noindent {\bf R6} The proposed path following controller is obtained with modifying the I\&I orbital stabilization technique proposed in \cite{ORTrnc}. The main difference relies on that the proposed target dynamics is dependent of the systems state $x$. Although the closed-loop dynamics is autonomous, we do not guarantee its orbital stability. Instead, the obtained orbit is almost periodic.

\vspace{0.3cm}

\noindent {\bf R7}
Assumption \ref{ass:h} ensures that the ``output" $\phi(x)$ has a well-defined vector relative degree $(2,2)$. Hence, the task {\bf T2} can be trivially satisfied with an input-output linearization scheme. As it is well-known that the latter operation is not robust, we leave open the task {\bf T2}, which may be accomplished with another controller.

\vspace{0.3cm}

\noindent {\bf R8}
In view of the remark above, Proposition \ref{prop:main} provides a \emph{constructive} solution to the PFP. That is, first, design a feedback law to stabilize the system \eqref{dyn_z}, where the controller is a function of $w(\xi,x)$, which is to be determined. Then, choose $w(\xi,x)$ to guarantee the boundedness of $(\yn, \dot{\bq}_{N})$. We underscore here that the free mapping $w(\xi,x)$ plays a similar role as the guidance laws \cite{AGUetal,FOS,PALetal} but, in contrast to them, the proposed design is a static feedback law.

\vspace{0.3cm}

\noindent {\bf R9}
If there are more than two inputs, {\em i.e.}, $m >2$, we may change the input as
$$
u = \alpha_{\tt u} (x) + \gamma_{\tt u}(x) u_c,
$$
with two variables of $u_c$ only controlling the off-the-manifold coordinate $z=\phi(x)$, and others affecting $(\yn,\dot{\bq}_{N})$. Such an idea is widely adopted in the literature on transverse feedback linearization, see \cite{CONetal} for an example.

\vspace{0.3cm}

\noindent {\bf R10}
In the transverse feedback linearization approach to the PFP \cite{CONetal,HLAetal,MORetal} guidance laws are not used. Instead, $\Phi(x)$ is selected as the transversal coordinate and the controller is obtained via feedback linearization. The main drawback of this approach is that $\nabla \Phi(x)$ appears in the input matrix, and then in the denominator of the feedback law---yielding an intrinsically local result. Moreover, to guarantee a required orthogonality condition, it is claimed that the tangential coordinates should be calculated via \emph{online optimization}.

\vspace{0.3cm}

\noindent {\bf R11}
In the transverse feedback linearization methods the tangential coordinates, calculated online, are determined by the given geometric path $\Phi(h(q))=0$. In contrast with this, the coordinate $q_N$ in the proposed method---equivalently to the mapping $N(q)$---is determined by the regulated output $\bq_y=h(q)$, but \emph{independent} of the given path.

\vspace{0.3cm}

\noindent {\bf R12}
The following comment, related with the possibility of extending our design design to high-dimensional paths, is in order. Let us assume that the target is to make the position $q\in\rea^3$ follow a three-dimensional Jordan-curve $\calc$ depicted by $\Phi^h_1(q)=0$ and $\Phi^h_2(q)=0$ for some mappings $\Phi_i$ ($i=1,2$). According to the definition of Jordan curve, the path $\calc$ is diffeomorphic to a unit circle, thus we are able to find functions $T_i(\cdot)$ such that $T_1^2(q) + T_2^2(q)=1$. The problem can be solved in our framework selecting the regulated output as $\bq_y=h(q) = \col(T_1(q), T_2(q))$.

%
\section{Examples}
\label{sec6}
%
In this section we present two application examples of the proposed solution to the PFP. The first one is a rather trivial linear system that is chosen to illustrate, in the simplest possible case, the procedure. The second example is the widely studied model of marine surface vessels.
\subsection{A motivating example}
\lab{subsec61}
%
Consider a 3-dof linear underactuated mechanical system \eqref{pH}, with $(q,p) \in \rea^3 \times \rea^3$, $A(q) = I_{3}$, constant $M>0$ and $R\ge 0$, potential energy $U(q) =0$, and constant full-rank $G \in \rea^{3\times2}$. The output is a linear combination of positions as
$$
\bq_y = C q,
$$
where $C \in \rea^{2\times 3}$ is full rank and we define $\bq_N = C^\bot q$. Assuming $\rank\{CM^{-1}G\} = 2$ we satisfy Assumption \ref{ass:h}.

We choose the target dynamics as \eqref{eq:target} with ${\calr = r I_{2}>0}$, $V_d(\xi)$ defined via \eqref{eq:Vd} and we consider the simple case with the mapping $w(\cdot)$ a function of $q$ only. We require that the desired path satisfies Assumption \ref{ass:jordan_curve}.

The next step in the I\&I procedure delineated above is to find the attractive and invariant set \eqref{me} with $\bq_y = \xi$. Thus, we fix the mapping $\pi$ as $\pi(\xi,\bfx) = \col(\xi, C^\bot q, *)$, where ``*'' represents the last three elements to be determined. Then solving the immersion condition \eqref{dot_bfx}, we get
$$
\begin{aligned}
z 
& = CM^{-1}p - L(q,w(q))
\end{aligned}
$$
with
$$
L(q,w):= \begmat{ -r\Phi(Cq) & { w(q)}\\ \\ -{w(q)} & -r \Phi(Cq) } \nabla\Phi(Cq).
$$
We have
$$
\begin{aligned}
\dot{z}
& = CM^{-1} ( - RM^{-1}p +  Gu) - \nabla L \begmat{M^{-1}p \\ \dot{w}}\\
& = CM^{-1} ( - RM^{-1}p +  Gu)  -  (\nabla_q^\top L + \nabla_w^\top L  \nabla w ) M^{-1}p .
\end{aligned}
$$
We can design the feedback linearizing controller as
\begequ
\label{controller:lti}
\begin{aligned}
u  = ( CM^{-1}G  ) ^{-1}
\big[(\nabla_q^\top L + \nabla_w^\top L  \nabla w + CM^{-1}R) M^{-1}p
 - kz
 \big]
\end{aligned}
\endequ
with $k>0$ to ensure
$$
\lim_{t\to\infty} z(t) =0
$$
exponentially fast.

Finally, we need to design $w(q)$ to guarantee the boundedness of $(\bq_N, \dot{\bq}_N)$. We have
\begequ
\label{ddotN}
\begin{aligned}
\ddot{\bq}_N & = C^\bot M^{-1} [ -RM^{-1}p +  G u] \\
& := \alpha_{\tt N}(x,w) + \beta_{\tt N}(x) \dot{w} + \epsilon_t \\
& := \alpha_{\tt N}(x,w) + \beta_{\tt N}(x) \nabla_a^\top L  \nabla w+ \epsilon_t
\end{aligned}
\endequ
with $\epsilon_t$ an exponentially decaying term caused by $z(0)$ and
$$
\begin{aligned}
\beta_{\tt N}(x,w) & :=
- C^\bot M^{-1} \big( R M^{-1} p
+ G  ( CM^{-1}G  ) ^{-1} \times \\
& \quad \quad (\nabla_q^\top L + CM^{-1}R )M^{-1}p \big)
\\
\alpha_{\tt N}(x,w) & :=
C^{\bot} M^{-1} G(CM^{-1}G)^{-1} \nabla_w L.
\end{aligned}
$$
In the dynamics \eqref{ddotN}, the function $w(q)$ can be assigned, which can be regarded an additional control to make $(\bq_N, \dot{\bq}_N)$ bounded. Its design needs to be studied case by case.


\vspace{0.3cm}

\noindent {\bf R13}
If we allow the mapping $w(\cdot)$ be a function of time as well, that is, $w(\xi,x,t)$, we are able to design a {\em dynamic} feedback law by regarding $\dot{w}$ in \eqref{ddotN} as an additional control input, making the construction more flexible. For such a case, the design of the integral action of $\dot w$ may resemble that of the projection variable $\theta$ in Section \ref{sec22}.

\vspace{0.3cm}

Let us consider a simple example with $M=I_{3}$, $R=\diag(0,0,R_3)$, $\calr = {1\over2} I_2$ and
$$
G = \begmat{ 1 & 0 \\ 0 & 1\\ 1 & 0},
~
C = \begmat{1 & 0 & 0\\ 0& 1& 0},
~
\Phi(q) = q_1^2 + q_2^2 - 1.
$$
We get $\dot{q} =p$ and
\begalis{
z&= \col(\dot{q}_1, \dot{q}_2) -L(q)\\
L(q,w) &:= - (q_1^2 + q_2^2 -1)\begmat{q_1 \\ q_2} + 2w \begmat{q_2 \\ -q_1}.
}
The obtained controller is
\begin{equation}
\label{controller:lti2}
  u = \nabla L^\top \dot q - kz.
\end{equation}
At last, we verify the boundedness of $(q_3,\dot{q}_3)$ under the feedback law \eqref{controller:lti2} with any non-zero constant $w$, the dynamics of which is
$$
\begin{aligned}
\dot{q}_3 &:= p_3
\\
\dot{p}_3  &= -R_3 p_3 + \nabla L_1(q_1,q_2) \begmat{\dot{q}_1\\ \dot{q}_2} - kz_1
\end{aligned}
$$
with $L_1(q_1,q_2)$ the first column of $L(q,w)$. Note that $(q_1,q_2)$ exponentially converges to \emph{periodic} signals with \emph{zero means}. It yields that $\nabla L_1(q_1,q_2)\col(\dot{q}_1,\dot{q}_2) = {d\over dt}{L}_1(q_1,q_2)$ also converges a periodic signal with zero mean. Note that the dynamics of $p_3$ is incrementally stable, thus $p_3$ exponentially converges to a periodic signal with zero mean. Using the properties of periodic functions, the anti-derivative function (primitive) of $p_3$ is a periodic signal, yielding the boundedness of $(q_3, \dot{q}_3)$.

Fig. \ref{fig:LTI} (the first three figures) shows the simulation results with the initial condition $q(0)=[2~0.5~1]^\top$, $p(0)= [0~0~1]^\top$, $w= \hal$ and $k=2$. We can observe that the system output converges to the desired path with bounded $(q_3,\dot q_3)$. We compare the proposed path following control with the I\&I orbital stabilization controller in \cite[Section 3.1]{ORTrnc}, namely,
\begin{equation}
\label{controller:lti3}
 u = - \mathbb{J}\col(q_1,q_2) - (\mathbb{J}+I_2) \col(\dot q_1, \dot q_2)
\end{equation}
with
$$
\mathbb{J} = \begmat{0 & 1\\ -1 &0}.
$$
It verifies a defect of the I\&I orbital stabilization controller \eqref{controller:lti3} in \cite{ORTrnc}, as well as appearing in the virtual holonomic constraint method---the steady-state behaviour depending on initial conditions, which is conspicuous by its absence in the proposed path following controller \eqref{controller:lti2}.

\begin{figure}[h]
  \centering
  \includegraphics[width=0.6\textwidth]{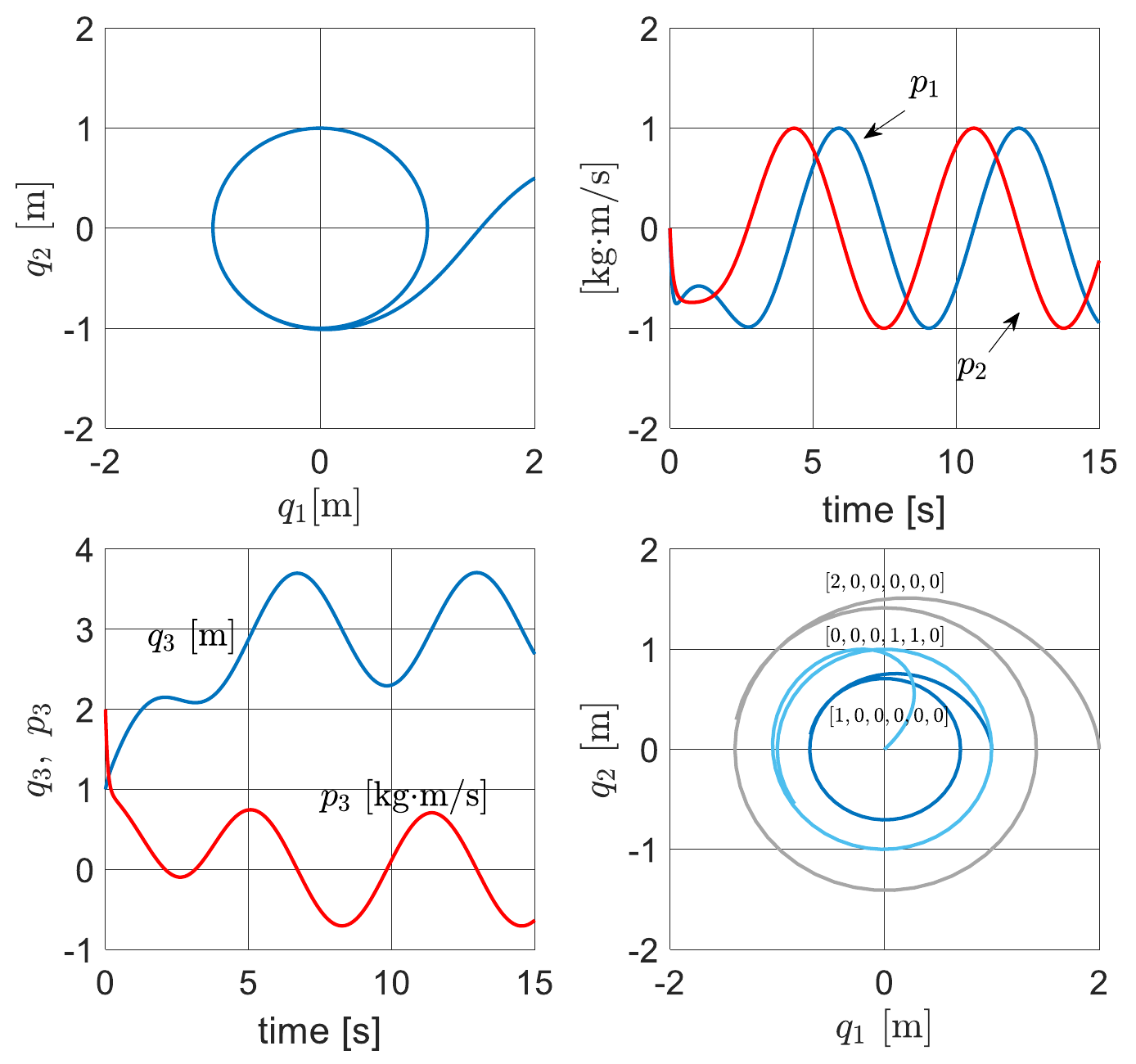}
  \caption{Performance of the path-following controller \eqref{controller:lti} (the first three figures) and the I\&I orbital stabilization controller \eqref{controller:lti3} in \cite{ORTrnc} with different intial conditions (the last figure) }\label{fig:LTI}
\end{figure}

%
\subsection{Marine Surface Vessels}
\lab{subsec62}
%
\subsubsection{Dynamical Model}
\label{sec:example1}

~

\vspace{0.1cm}

The dynamical model of marine surface vessels is given by \eqref{pH}, where $q \in \rea^2 \times \mathbb{S}$ denotes the generalized-position vector and yaw angle in the earth-fixed coordinate. The potential energy $U(q)$ equals to zero, and the matrix $M$ is constant in the body-fixed coordinate. The mappings of  \eqref{pH} are defined as\footnotetext{We consider the case of linear damping. Here, $C(x)$ is the Coriolis-centripetal matrix, and the positive definite matrix $D(x)$ is the damping matrix related with hydrodynamics. We lump these matrices into $R(x)$, which verifies $R(x) + R^\top(x) >0$.} \cite{DONPER}
$$
R(x) = C(q, M^{-1}p) + D,
$$
and
\begequ
\label{MCD}
\begin{aligned}
M & = \begmat{ m_{11} & 0 & 0 \\ 0 & m_{22} & m_{23} \\ 0 & m_{32} & m_{33}  },
~~
D  = \begmat{d_{11} & 0 & 0 \\ 0 & d_{22} & d_{23} \\ 0 & d_{32} & d_{33}}
\\
C(q,v) & = \begmat{ 0 & 0 & -m_{22}v_2 - m_{23} v_3 \\
0 & 0 & m_{11} v_1 \\
m_{22} v_2 + m_{23}v_3 & - m_{11}v_1 & 0
}.
\end{aligned}
\endequ
with ${v} {= M^{-1}p}$ physically representing the velocity in the body frame. The control input $u \in \rea^2$ is generated by the actuators---propellers and rudders. The mapping $A(q) \in {\rm SO(3)}$ is the rotation matrix
\begequ
\label{A}
\begin{aligned}
A(q) & = \begmat{\cos (q_3) & - \sin (q_3) & 0 \\ \sin (q_3) & \cos (q_3) & 0 \\ 0 & 0 & 1}
\\
& := \begmat{\mathbf{A}(q) & 0_{2\times 1} \\ 0_{1\times 2} & 1}.
\end{aligned}
\endequ
 We now take some practical modelling aspects into consideration and make the following assumption \cite{FOS} implying that the position state $q_2$ is underactuated.
\begin{assumption}
\rm\label{ass:b22}
The body-fixed coordinate frame is located at the pivot point such that $M^{-1}Gu = [u_1,0, u_2]^\top$.
\end{assumption}
This is  a standard assumption to ensure, via a suitable selection of the body-fixed coordinate, that $u$ is normalized.

\subsubsection{Constructive Solution}
\label{sec:522}

~

\vspace{0.1cm}

{To circumvent the difficulty of underactuation of $q_2$}, we follow the \emph{hand position} method in \cite{PALetal} by selecting a certain point on the center line of the vessel, then performing the partial change of coordinates \eqref{defbfq}, which in this example takes the form with a free parameter $\ell>0$

\begequ
\label{y}
T(q) = \begmat{~q_1 + \ell \cos(q_3)~ \\~ q_2 + \ell \sin(q_3)~ \\ q_3},
\endequ
clearly satisfying $\rank\{\nabla T(q)\} =3$. {According to Proposition \ref{prop:main}, our task is to design a feedback law guaranteeing {\bf T1} and {\bf T2} for the closed-loop dynamics. We summarise the results of the path following controller design as follows.}
\begin{proposition}
\label{prop:vessel}\rm
Consider the model \eqref{pH} with \eqref{MCD} and \eqref{A}, verifying {Assumption \ref{ass:b22}}. Select $\ell$ such that the inequality
\begin{equation}
\label{ineq:ell}
  \ell > - {d_{33}m_{23} - d_{23}m_{33} \over d_{22}m_{33} - d_{32}m_{23} }
\end{equation}
holds. Then, for any smooth path
$
\{\bq_y \in \rea^2 ~|~\Phi(\bq_y) =0\}
$
satisfying {Assumption \ref{ass:jordan_curve}} with \eqref{y}, there exists a constant $w_0>0$ such that for any $0<|w(\xi,x)|<w_0 $,  the feedback law\footnotetext{The time derivative ${d\over dt}(F(\bq_y,x) \nabla V_d(\bq_y))$, which can be written as a mapping of the system states, is well-posed.}
\begequ
\label{control_vessel}
u = g_z^{-1}(x) \big[
- f_z(x) + {d\over dt} (F(\bq_y,x) \nabla V_d(\bq_y)) - kz
\big]
\endequ
with $k>0$, the mappings $V_d$ and $F$ defined as \eqref{eq:Vd} and \eqref{F}, respectively, the mappings $g_z$ and $f_z$ defined in Appendix, and $z= \phi(x)$ defined as \eqref{app:phi_x}, achieves the path following problem.
\end{proposition}
\begin{proof}
\rm
According to Proposition \ref{prop:main}, we split the proof into two parts, verifying the conditions {\bf T1} and {\bf T2}: 1) the convergence of the off-manifold coordinate and 2) the boundedness of $(\yn,\dot{\bq}_{N})$.

1) The off-manifold coordinate $z=\phi(x)$ is defined in \eqref{app:phi_x}, the dynamics of which is
$$
\dot{z} = f_z(x) + g_z(x) u - {d\over dt}(F(\bq_y,x) \nabla V_d(\bq_y))
$$
with $g_z(\cdot)$ and $f_z(\cdot)$ given in Appendix B. Under the proposed controller, we get
$$
\dot{z} = - kz,
$$
thus $|z(t)|$ converging to zero exponentially. We underscore that the feedback law in \eqref{control_vessel} is a function of $\bq_y, ~\dot{\bq}_y$ and $\yn=q_3$, but independent of $\dot{q}_{3}$. Due to $q_3 \in \mathbb{S}$, we conclude $u \in \call_\infty$.

2) By simply selecting $N(q) = q_3$, we guarantee the injectivity of $T(q)$ defined in \eqref{defbfq}, thus $\bq_N = q_3$. According to Proposition \ref{prop:main}, we further need to verify the boundedness of the states $(q_3,\dot{q}_3)$, the dynamics of which is
\begequ
\label{reduced-dyn}
\begin{aligned}
\dot{q}_{3} & = v_{3} \\
\dot{v}_{3} & = f_{\tt N}(q_3, v_{3}, \bq_y, \dot{\bq}_y)
\end{aligned}
\endequ
with
$$
\begin{aligned}
& f_{\tt N}(q_3, v_3, \bq_y, \dot{\bq}_y)
\\
& \hspace{0.4cm} :=
- \bigg[ \bigg(\delta_3 - {\delta_1-1\over \ell}\bigg)|\dot{\bq}_y|^{1\over 2} \cos(q_3 - \psi) + \delta_4 + {\delta_2 \over\ell} \bigg]v_{3}\\
 & \hspace{0.9cm}   - \bigg[ {\delta_3 \over \ell} |\dot{\bq}_y|^{1\over 2}  \cos(q_3-\psi) +{\delta_4\over \ell}\bigg] |\dot{\bq}_y|^{1\over2} \sin(q_3 - \psi)
 + {1\over \ell}[g_z (x) u + f_z(x)] \begmat{-\sin(q_3) \\ \cos(q_3)} \\
& \hspace{0.2cm} \psi  := \text{atan2}(\dot{\bq}_{y1}, \dot{\bq}_{y2}).
\end{aligned}
$$

The boundedness of $q_3$ is obvious due to $q_3 \in \mathbb{S}$. Hence, we only need to prove the boundedness of $\dot{q}_3$. If the system states are on the manifold with $z=0$, then
$$
\begin{aligned}
|\dot{\bq}_y|^{1\over 2} = |w(\bq_y,x)| |\nabla \Phi| \le w_0 \sup_{\Phi(\bq) = 0}|\nabla \Phi| =: w_1
\end{aligned}
$$
is bounded. Since the origin of the dynamics of $z$ is exponentially stable, for any $\varepsilon>0$ there always exists a time instant $T_1>0$ such that
\begequ
\label{bqy<=}
|\dot{\bq}_y(t)|^\hal \le w_1 + \varepsilon, \quad \forall t\ge T_1.
\endequ
We can rewrite the dynamics \eqref{reduced-dyn} of $v_{3}$ in the form
\begequ
\label{v_w}
\dot{v}_{3} = -K(\dot{\bq}_y, q_3) v_{3} - \Delta(\bq_y,\dot{\bq}_y, q_3, v_3)
\endequ
with
$$
\begin{aligned}
K(\dot{\bq}_y, \yn)
& = \bigg(\delta_3 - {\delta_1-1\over \ell}\bigg)|\dot{\bq}_y|^{1\over 2} \cos(q_3 - \psi) + \delta_4 + {\delta_2 \over\ell} \\
 \Delta(\bq,\dot{\bq}_y, q_3, v_3) & =- \bigg[ {\delta_3 \over \ell} |\dot{\bq}_y|^{1\over 2}  \cos(q_3-\psi) +{\delta_4\over \ell}\bigg] |\dot{\bq}_y|^{1\over2} \sin(q_3 - \psi) \\
 & \hspace{0.5cm}
 + {1\over \ell}[g_z (x) u(\bq_y,\dot{\bq}_y,q_3) + f_z(x)] \begmat{-\sin(q_3) \\ \cos(q_3)}  .\\
\end{aligned}
$$

It is clear that $K(\dot{\bq}_y, q_3)$ is bounded, since $|\dot{\bq}_y| \in \call_\infty$ and $|\cos(q_3-\psi)| \le 1$. The inequality \eqref{ineq:ell} guarantees
$$
\delta_4 + {\delta_2 \over\ell} >0,
$$
which is the last two terms of $K(\cdot)$. Invoking \eqref{bqy<=} and $|\cos(q_3-\psi) |\le 1$, we can guarantee, by selecting small $w_0$ and small $\varepsilon$, that
$$
K(\dot{\bq}_y(t), q_3(t)) >0, \quad \forall t \ge T_1.
$$

The scalar ``time-varying'' system \eqref{v_w} is linear in $v_3$, thus forward complete for $ t \ge 0$. After the moment $t= T_1$, this system is input-to-state stable, which together with the fact $\Delta(\cdot) \in \call_\infty$, yields
$
v_{3}(t) \in \call_\infty,
$
completing the proof.
\end{proof}

\noindent {\bf R14}
From a pragmatic viewpoint, we should consider environmental disturbances on vessels, for instance, constant ocean currents. For this case, the kinematic model becomes
$$
\dot{q} = A(q)\nabla_p H(x) + \col(V_c, 0)
$$
with $V_c\in \rea^2$ an unknown constant vector. We may use the integral control \cite{DONJUN,ORTROM}
\begequ
\label{control_vessel2}
\begin{aligned}
u & = g_\ell^{-1}(x) \Bigg[
- f_\ell(x) + {\partial \kappa \over \partial \bq_y}{\partial h \over \partial q}A(q)M^{-1}p
+ {\partial \kappa \over \partial \bq_y}{\partial h\over\partial q_{12}}\theta
-\dot{\theta} -  k_\mathtt{p}z
\Bigg]
\\
\dot{\theta} & =  k_\mathtt{I} \nabla V_d(\bq_y) +  k_\mathtt{I} {\partial \kappa \over \partial \bq_y}{\partial h\over\partial q_{12}}z,
\end{aligned}
\endequ
to solve the PFP, where $q_{12}=\col(q_1,q_2)$ and\footnotetext{Consider the simple case with $w(\cdot)$ only a function of $\xi$.}
$$
z = \nabla h^\top AM^{-1}p - F(h(q))\nabla V_d(h(q)) +\theta,
$$
$k_{\tt I}>0$, and $k_{\tt p}> {1\over 4}w$. Due to the page limitation, we will report the theoretical analysis with disturbance attenuation in other place.

\subsection{Simulation Results}
\label{sec5.4}

In this section, we give the simulation results corresponding to Subsections \ref{sec:522}, using the model of a supply vessel in \cite{CAH}. The desired path is a circle defined by $\Phi(\bq_y) = |\bq_y|^2 - 10^4 =0$, where the regulated output is given by \eqref{y} with $\ell =18$. Since the model does not satisfy the Assumption \ref{ass:b22}, we give the normalized (dimensionless) input $u = M^{-1}B\tau$ with $\tau$ the actual surge thrust and rudder angle. The initial states are selected as $q(0)= [120~~ -90~~0]^\top $ and $p(0)= {0}_3$.

We first consider the case in the absence of ocean currents, with the controller \eqref{control_vessel}, with $k=0.1$, constants $R=10^{-5}$, and $a=0.04$. The performance is given in Fig. \ref{fig:1}, verifying the theoretical analysis in Proposition \ref{prop:vessel}. We then test the controller \eqref{control_vessel} with the ocean current $V_{\tt c}=[2~~1]^\top$ in Fig. \ref{fig:2}, in which the ultimate path is in the neighbourhood of the desired one. This is because exponential orbital stability enjoys some robustness in the presence of bounded perturbation, but reducing the steady-state accuracy.

\begin{figure}[h]
  \centering
  \includegraphics[width=0.6\textwidth]{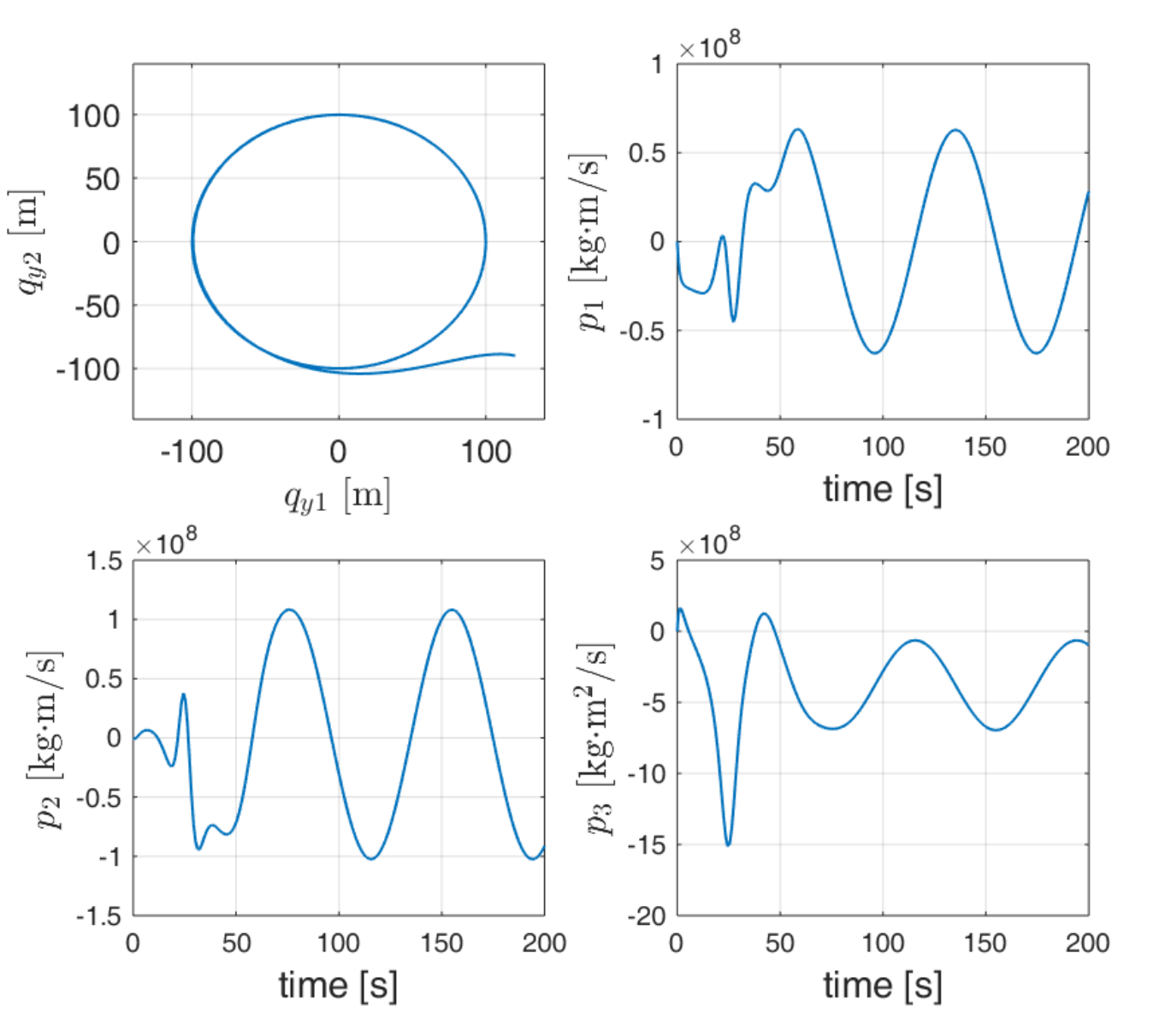}
  \caption{Performance with the controller \eqref{control_vessel} in the absence of ocean currents}\label{fig:1}
\end{figure}

\begin{figure}[h]
  \centering
  \includegraphics[width=0.6\textwidth]{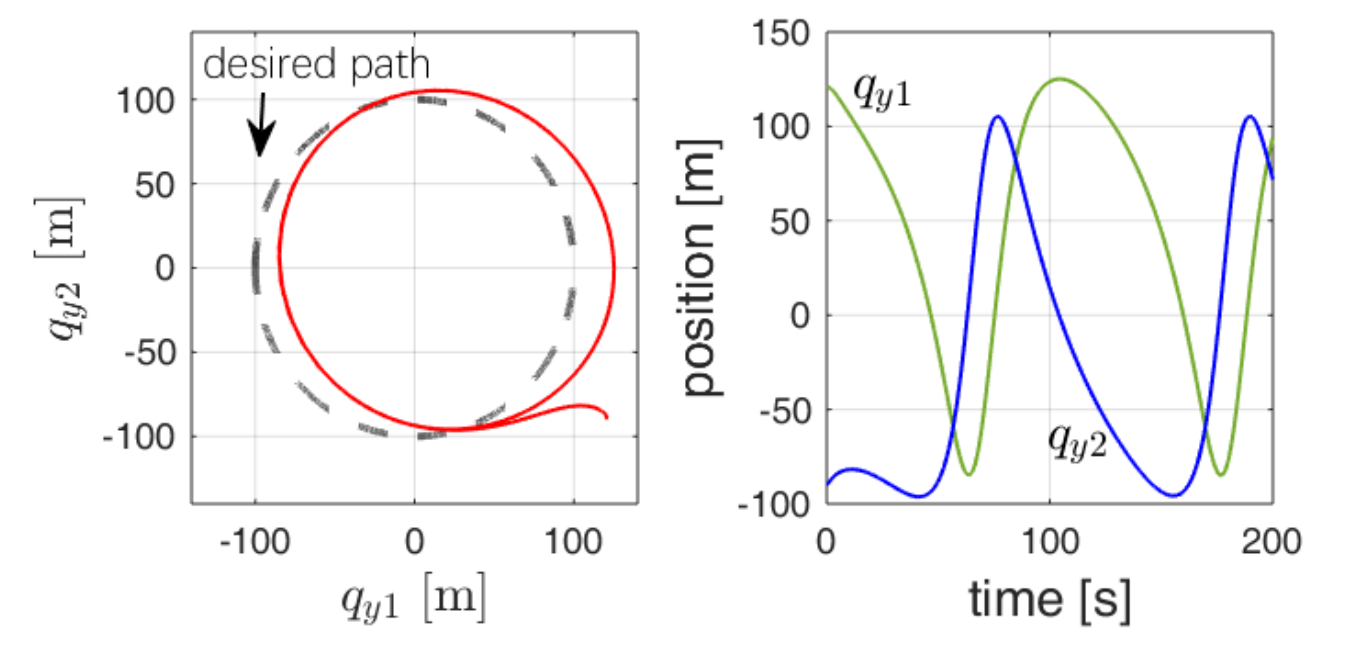}
  \caption{Performance with the controller \eqref{control_vessel} in the presence of ocean currents $V_\mathtt{c}= [2~1]^\top$}\label{fig:2}
\end{figure}

Next, we consider the integral control \eqref{control_vessel2} in the presence of ocean currents $V_{\tt c}=[5~~1]^\top$, with the parameters $k_{\tt p}=0.1,~\theta(0) = [0~~0]^\top$ and $k_{\tt I} = 5\times 10^{-7}$, the simulation results of which are shown in Fig. \ref{fig:3}. We observe that the trajectory ultimately converges to the desired path again without steady-state error.

\begin{figure}[h]
  \centering
  \includegraphics[width=0.6\textwidth]{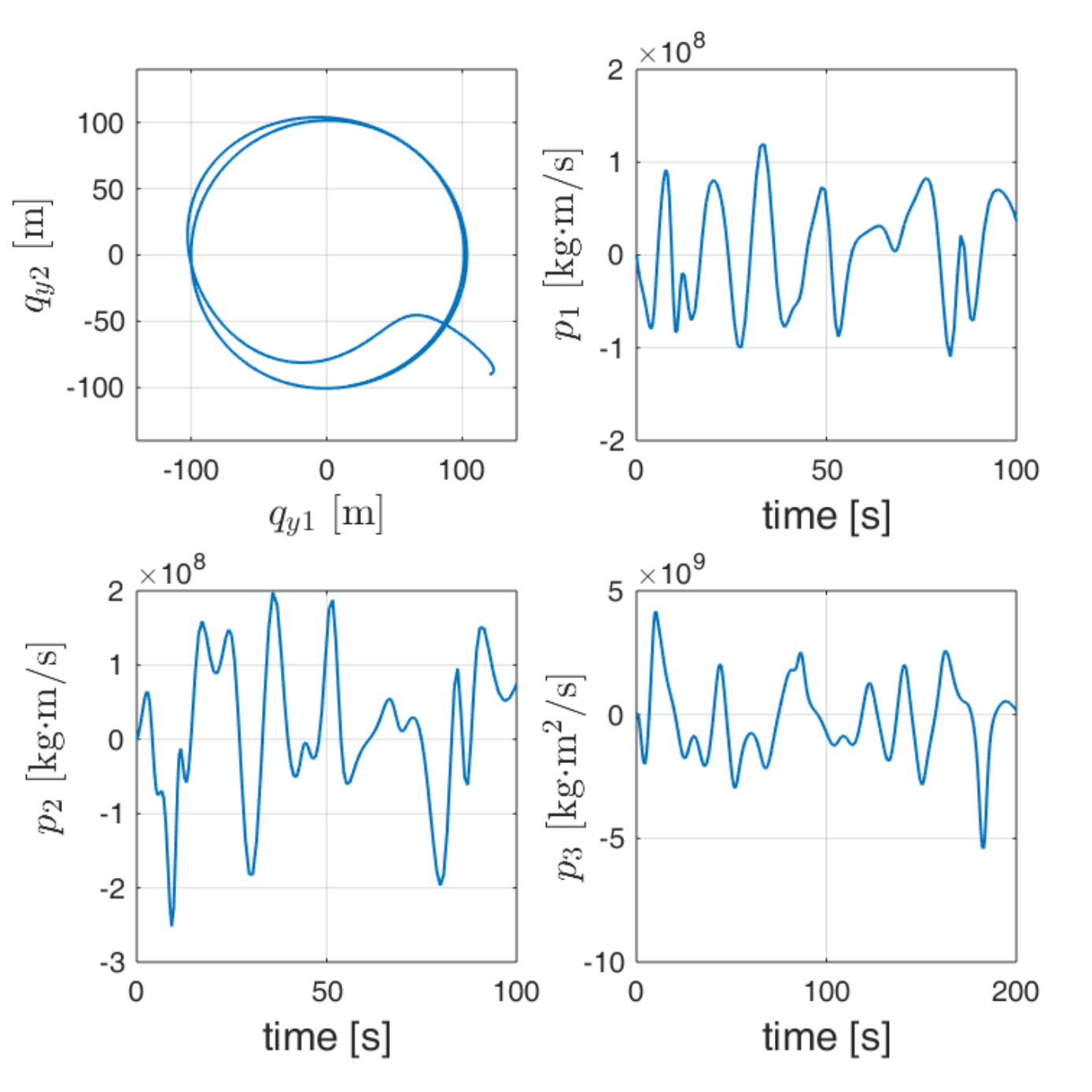}
  \caption{Performance with the dynamic controller \eqref{control_vessel2} in the presence of ocean currents $V_\mathtt{c}= [5~1]^\top$}\label{fig:3}
\end{figure}

%
\section{Concluding Remarks}
\label{sec7}
%
In this paper, we address the PFP of a class of underactuated mechanical systems, which is reformulated as orbital stabilization. We first adopt the Mexican sombrero energy assignment method in \cite{YIetal2} to design a target dynamics with an almost globally attractive limit cycle, {\em i.e.}, the desired path. Unlike \cite{YIetal2}, we do not assume the function $\Phi(\cdot)$ has a unique critical point. Then, we use the I\&I method \cite{ASTORT,ORTrnc} to achieve the task of path following.

The method includes the selection of the mapping $w(\xi,x)$ that can be used to shape the steady state behaviour. We also show, via simulations, that the proposed method can be further robustified by introducing an integral action. The theoretical analysis of the integral action is omitted here for brevity, and will be reported later. Finally, the proposed method is used to design a path following controller for underactuated surface vessels.

The extensions in the following directions are of interests.

\begin{itemize}
\setlength{\itemsep}{5pt}
  \item By adding an suitable integral action to the target oscillator, we are able to attenuate the unmatched external disturbances. More specifically for the marine vessel example, the proposed method is efficient to deal with the unmatched ocean currents.

  \item It is interesting to study the function design of $w(\xi,x)$ to achieve performance enhancement, {\em i.e.}, making $w(\xi,x)$ involve the functions of guidance laws. It is also promising to consider a time-varying $w(\xi,x,t)$ or an adaptation of $\dot{w}$ to get further performance enhancement or simplify the design.
  \item In this paper, we limit our attentions to the planar path, though we have demonstrated the possibility to deal with high-dimensional paths in {\bf R12}, {\em e.g.}, unmanned underwater vehicles (UUV), and robotics.
  \item In some cases, disturbances to marine surface vessels are more complicated, which may be modelled as constant ocean currents with harmonic (first and second-orders) signals. It is reasonable to use finite parameters, including amplitudes, frequencies and phases, to parameterize environmental disturbances, and then adaptively achieve the path following.
\end{itemize}

\appendix

\section*{A. Proof of Proposition \ref{pro3}}
\label{sec:app_proof_lemma}


\noindent {\bf Verification of {\bf A2'}.} In the case when the mapping $\alpha(\xi,\bfx)$ depends on $x$, the immersion condition  \eqref{dot_bfx} takes the form
\begali{
[I_{2n} - \nabla_\bfx \pi^\top(\xi,\bfx)]\big(f(\pi(\xi,\bfx) + g(\pi(\xi,\bfx)) c(\xi,\bfx) \big)
 = \nabla_\xi \pi^\top(\xi,\bfx) \alpha(\xi,\bfx).
\label{fbi2}
}
To obtain the invariant manifold $\calm$ defined in \eqref{eq:M}, we select the first $n$ equations of $\pi(\xi,\bfx)$, denoted as $\pi_{(1,n)}(\cdot)$, are
$$
\begmat{\pi_1(\xi,\bfx) ~\big|&~ \ldots ~\big| & ~\pi_n(\xi,\bfx) }
=
\begmat{ \xi^\top ~&~ \bq_N^\top}
$$
Correspondingly, the first $n$ equations of \eqref{dot_bfx} are $\dot{\bq} = \col(\dot{\xi}, \dot{\bq}_N)$, which can be guaranteed selecting
\begequ
\label{phix}
\pi(\bfx,\xi) =
\left[
\begin{array}{c}
\xi \\ \yn \\
\\
M(T^{\tt I}) \cala^{-1}
\begmat{F(\xi,\bfx)\nabla V_d(\xi) \\ \dotbfqn}
\end{array}
\right]
 \in \rea^{2n}.
\endequ
However, the above equation is only a necessary condition to \eqref{dot_bfx}. We further need to prove the existence of $c(\cdot)$ which guarantees the last $n$ equations of \eqref{dot_bfx}.

We look at the right-hand side of the remaining $n$ equations in \eqref{dot_bfx}, which is
\begalis{
  \dot{\pi}_{(1,n)} 
& =   \sigma(\bfx)
+  M \cala^{-1}  \begmat{{d\over dt}(F\nabla V_d) \\ \ddot{\bq}_N},
}
where we have defined
$$
\sigma(\bfx) := \Big(\sum_{i=1}^{n} (\nabla_{\bq_i}(M\cala^{-1}) )  \cala M^{-1} p e_i \Big)
\cala M^{-1}p
$$
and $e_i$ is the $i$-th Euclidean basis vector.

The left-hand side of the remaining $n$ equations in \eqref{dot_bfx} is
$$
  \dot{p}|_{\bfx = \pi(\xi,\bfx)} = - \cala^\top(\bq) \nabla_\bq \calh - \mathbf{R} M^{-1}p + G(\bq)c(\cdot)
$$
by applying $u=c(\xi,\bfx)$ and using \eqref{pH_trans}. Then, we only need to find $c(\xi,\bfx)$ such that
$
  \dot{p}|_{\bfx = \pi(\xi,\bfx)}  =    \dot{\pi}_{(1,n)}
$
or equivalently,
\begin{equation}\label{eq:dot-q}
\begin{aligned}
   \begmat{{d\over dt}(F\nabla V_d) \\ \ddot{\bq}_N} + \cala^{-1}M
\Big(
    \cala^\top \nabla_\bq \calh + \mathbf{R}M^{-1}p + \sigma(\bfx)
\Big)
=
\cala M^{-1} G(q) c(\xi,\bfx).
\end{aligned}
\end{equation}
Multiplying the full-rank matrix $I_{n\times n} = \col(\mathbb{I} , \mathbb{I}^\bot)$ to the both sides, we get the first 2 rows and the last $(n-2)$ rows of \eqref{eq:dot-q}.
%
The first 2 rows can be used to obtain the controller the mapping $c(\cdot)$ given by
$$
\begin{aligned}
c(\xi,\bfx)   =
[\nabla h^\top A M^{-1}G ]^\dagger
%
\Bigg[
{d\over dt}(F\nabla V_d)
 +
\mathbb{I} \cala^{-1}M
\Big(
\cala^\top \nabla_\bq \calh
+ {\bf R} M^{-1}p +
\sigma(\bfx)
\Big)
\Bigg]
\end{aligned}
$$
where $\mathbb{I} := [I_{2\times 2} ~~0_{2\times(n-2)}]$, $\yn$ is defined in \eqref{defbfq} and observing that $ {\mathbb{I}\cala M^{-1}G = [\nabla h^\top AM^{-1}G ]}$ is a full-rank fat matrix, due to Assumption \ref{ass:h}---hence, its pseudoinverse is well defined, and {we have also used the relations
 $
 \cala = \nabla T^\top A
 $
 and $\mathbb{I} \nabla T^\top = \nabla h^\top$. }The remaining equations multiplied by $\mathbb{I}^\bot$ is nothing, but just ${d\over dt} (\dot{\bq}_N) =\ddot{\bq}_N$, which is true automatically.

\noindent {\bf Verification of {\bf A3'}.} The implicit manifold condition {\bf A3'} in the original coordinate $x=\col(q,p)$ of the dynamics \eqref{pH} takes the form
\begequ
\label{phi(x)}
\phi(x) =0 \quad \Leftrightarrow \quad
\bfx = \pi(\bfx,\xi).
\endequ
Some simple calculations prove that the equivalence holds with
\begequ
\label{app:phi_x}
\phi(x) = \nabla h^\top AM^{-1}p - F(h(q),x) \nabla V_d(h(q)).
\endequ

\section*{B. Mappings in the vessel example}
\label{sec:mappings_ship}

We give the mappings in the example of marine surface vessels as follows \cite{PALetal}.
\begequ
\label{mappings_vessel}
\begin{aligned}
  g_z(x) & = \begmat{\cos(q_3) & - \ell \sin(q_3) \\ \sin(q_3) & \ell \cos(q_3) },
  \quad
  f_z(x) = \mathbf{A}(x) \begmat{F_1(x) - v_2v_3 - \ell v_3^2 \\ v_1v_3 + F_2(x)v_3 + F_3(x)v_2 + F_4(x)\ell}
\\
  m_0 & = m_{22}m_{33} - m_{23}^2
\\
  F_1(x) & = {1\over m_{11}} (m_{22}v_2 + m_{23}v_3) v_3  - {d_{11} \over m_{11}}v_1
\\
  F_2(x) & = {1\over m_0} \Big( ( m_{23}^2 - m_{11}m_{33} )v_1 + (d_{33}m_{23} - d_{23}m_{33})     \Big)
\\
  F_3(x) & = {1\over m_0} \Big( (m_{22} - m_{11}) m_{23}v_1 - (d_{22}m_{33} - d_{32}m_{23})  \Big)
\\
  F_4(x) & = {1\over m_0} \Big( \big( m_{23}d_{22} - m_{22}(d_{32} + (m_{22} - m_{11})v_1 \big)\big)v_2 \\
    & \quad\quad
    + \big( m_{23}(d_{23} + m_{11}v_1) - m_{22}(d_{33} + m_{23}v_1) \big)v_3 \Big)
\end{aligned}
\endequ
with
$
{v  = M^{-1}p}.
$
Some parameters are defined as
$
\delta_1  := {1\over m_0} (m_{11}m_{33}- m_{23}^2), ~
 \delta_2  := {1\over m_0} (d_{33}m_{23} - d_{23}m_{33}),~
\delta_3  := {1\over m_0} (m_{11} - m_{22})m_{23},
$
and
$
 \delta_4  := {1\over m_0} (d_{22}m_{33} - d_{32}m_{23}).
$

\end{document}